\documentclass[conference]{IEEEtran}

\usepackage{cite}
\usepackage{graphicx}
\usepackage{epstopdf}
\usepackage{epsfig}
\usepackage[latin1]{inputenc}
\usepackage{amsthm}
\usepackage{subfigure}

\theoremstyle{plain}
\newtheorem{theorem}{Theorem}

\theoremstyle{remark}
\newtheorem{lemma}{Lemma}

\begin{document}

\title{The Price of Updating the Control Plane in Information-Centric Networks}

\author{\,}

\author{Bita Azimdoost$^{\dagger}$, Cedric Westphal$^\ddagger$$^*$, and Hamid R. Sadjadpour$^\dagger$\\
$^\dagger$Department of Electrical and $^\ddagger$Computer Engineering, 
University of California Santa Cruz\\
$^*$ Huawei Innovation Center, Santa Clara, CA
}

\maketitle

\begin{abstract}
We are studying some fundamental properties of the interface between control and data planes in Information-Centric Networks.
 We try to evaluate the traffic between these two planes based on allowing a minimum level of acceptable distortion in the network state representation in the control plane. We apply our framework to content distribution, and see how we can compute the overhead of maintaining the location of content in the control plane. This is of importance to evaluate content-oriented network architectures: we identify scenarios where the cost of updating the control plane for content routing overwhelms the benefit of fetching a nearby copy. We also show how to minimize the cost of this overhead when associating costs to peering traffic and to internal traffic for operator-driven CDNs.
\end{abstract}

\section{Introduction}

Routing in Information-Centric Networks (ICNs) is performed based upon content names. Name Resolution Service-based ICNs \cite{Ahlgren2012Survey} require the routing mechanism to be dynamically updated with the content location. Content location from the forwarding plane needs to be delivered to the control plane. This raises the following question: depending on the size of the domain being controlled, of the underlying state space, of the dynamics of the evolution of the state in the forwarding plane, how much data is required to keep the ICN control plane up to date of the content location?

We consider the issue of maintaining a consistent view of the underlying state at the control layer, through an abstracted mechanism, which can be applied to a wide range of scenarios: we consider the underlying state as an evolving random  process, and calculate the information theoretic rate that this process would create to keep the representation of this state up-to-date in the control plane. This provides a lower bound on the bandwidth overhead required for the control plane to have an accurate view of the forwarding plane.

We then illustrate the power of our model by focusing on the specific case of locating content in ICNs. Enabling content routing has attracted a lot of attention recently, and thus we are able to shed some light on its feasibility.  In this case, the underlying state depends on the size and number of caches, on the request for content process and on the caching policy. We  apply our framework to derive the bandwidth needed to accurately locate a specific piece of content. We observe that there is a  trade-off for keeping an up-to-date view of the network at the cost of  significant  bandwidth utilization, versus the gain achieved by fetching the nearest copy of the content. We consider some simple scenarios to illustrate this trade-off.

Our contribution is as follows:\\
$\circ$ A framework is presented to quantify the minimal amount of information required to keep a (logical) control plane aware of the state of the forwarding plane. We believe this framework to be useful in many distributed systems contexts (Lemma \ref{lem:01});\\
$\circ$ This framework is applied to the specific case of locating content, and the effect of the availability of caches, the caching policy and the content popularity on content location is studied. We can thus apply our results to some of the content-oriented architectures (Theorems \ref{thm:1}, \ref{thm:2});\\
$\circ$ We see how our framework allows to define some optimal policies with respect to the content that should be cached for an Information-Centric Network (Section\ref{sec:costanalysis}).

We quickly note that our framework does not debate the merit of centralized vs distributed, as the control layer we consider could be either. For a routing example, our model would provide a lower-bound estimate of the bandwidth for, say OpenFlow to update a centralized SDN controller, or for a BGP-like mechanism to update distributed routing instances.

Our results are theoretic in nature, and provide a lower bound on the overhead. We hope they will provide a practical guideline for protocol designers to optimize the protocols which synchronize the network state and the control plane.

The rest of the paper is organized as follows. After going over some related work in section \ref{sec:related}, we  introduce our framework to model the protocol overhead in section \ref{sec:protocol}. Section \ref{sec:scenarios} utilizes the derived model to study some simple caching networks. We show the power of the model in the protocol design by computing the cost of content routing in Section~\ref{sec:costanalysis} and suggesting a cache management policy. 
Finally, section \ref{sec:conclusion} concludes the paper.

\section{Related Work}
\label{sec:related}

SDN makes the separation explicit between the control and forwarding layer. Interactions between the control and forwarding planes has been pointed out as one of the bottlenecks of OpenFlow~\cite{McKeown2008OpenFlow}. As a consequence, \cite{Curtis2011Devo} or \cite{Yu2010Difane} attempt to reduce the amount of interaction in between the switches and the control layer. \cite{Levin2012Logically} studies the gap between the state of the actual system and the view of the (logically centralized) controller, focusing on consistency. There has been no attempt to model the interaction between the control and forwarding layers to our knowledge.

The control plane needs to obtain adequate information about the underlying states so that the network can perform within a satisfactory range of distortion. The first theoretical study of this information was conducted by Gallager in \cite{Gallager1976Basic}, which utilized rate distortion theory to calculate the information required to extract network parameters. 
\cite{Wang2012Cost} applied these ideas to mobile wireless networks. An information-theoretic framework to model the relationship between network information and network performance was derived in \cite{Hong2009Impact}.

One impetus to study the relationship between the control layer and the network layer comes from the increased network state complexity from trying to route directly to content. Request-routing mechanisms have been in place for a while\cite{Barbir2003Requestrouting} and proposals~\cite{Davie2012Framework} have been suggested to share information between different CDNs, in essence enabling the control planes of two domains to interact. Many architectures have been proposed that are oriented around content\cite{Gritter2001Architecture,Koponen2007Dataoriented,Jacobson2009Networking,Zhang2010Named,Pursuit,Ahlgren2012Survey} and some have raised concerns about the scalability of properly identifying the location of up to $10^{15}$ pieces of content\cite{Ghodsi2011InformationCentric}. Our model presents a mathematical foundation to study the pros and cons of such architecture.

Cache management goes jointly with content routing. \cite{Tang2008BenefitBased}\cite{Bhattacharjee1998Self}\cite{Cho2012WAVE} present cache management policies. Some cooperative cache management algorithms have been developed in \cite{Borst2010Distributed} which attempt to maximize the traffic volume served from cache and minimize the bandwidth cost in content distribution networks. \cite{Sourlas2012Autonomic} proposes some online cache management algorithms for Information Centric Networks (ICNs) where all the contents are available by caching in the network. \cite{Chai2012Cache} investigates if caching only in a subset of nodes along the path in ICNs can achieve better performance in terms of cache hit rate. None of these work study the overhead required to learn the content location. 

\section{Protocol Overhead Model}
\label{sec:protocol}

We now turn our attention to the mechanism to synchronize the view at the control layer with the underlying network state. Assume that $S_X(t)$ describes the state of random process $X$ in a network at time $t$. 

In order to update the control plane's information about the states of $X$ in the network, the forwarding plane must send update packets regarding those states to the control plane whenever some change occurs. Let $\hat{S}_X(t)$ denote the control plane's perceived state of $X$ at time $t$. It is obvious that no change in $\hat{S}_X$ will happen before $S_X$ changes, and if $S_X$ changes, the Control plane may or may not be notified of that change. Therefore,  there are some instances of time where $\hat{S}_X\ne S_X$. 

In this paper, we consider that the state can have two values $'0'$ and $'1'$. For instance, a link can be up or down; or a piece of content can be present at a node, or not. 

Please note: It is easy to see that such boolean state space can be generalized to other possible values for $S_X$. For instance, if one wanted to measure the congestion on a link, one could quantize the link congestion into bins (say bins $b_1$ to $b_{10}$ for normalized link utilization between 0 and 0.1, 0.1 to 0.2, $\ldots$, 0.9 to 1) and map a the link utilization to a 0-1 variable such that $b_i = 1$ if the current link utilization is in $((i-1)/10,i/10)$ and 0 otherwise. Therefore, there is no loss of generality of selecting a 0-1 variable, and it greatly simplifies the exposition of the results. 

Let $\{Y_m\}_{m=1}^{\infty}$ and $\{Z_m\}_{m=1}^{\infty}$ denote the sequences of $'0'$s and $'1'$s time durations of $S_X(t)$ respectively, and $\{T_m\}_{m=1}^{\infty}$ denote the times of changes. We assume that $Y_m$ is an i.i.d sequence with probability density function (pdf) $f_Y(y)$ and mean $\theta_X$, and $Z_m$ is another i.i.d. sequence with pdf $f_Z(z)$ and mean $\tau_X$. We also assume that any two $Y_m$ and $Z_m$ are mutually independent\footnote{There is also no loss of generality in assuming independence of these processes for the following reason: we consider large distributed systems, where the input is driven by a large population of users (smaller systems offer no difficulty in tracking in the control plane what is happening in the data plane). It is a well known result that the aggregated process resulting from a large population of uncoordinated users will converge to a Poisson process, and therefore the events in the future are independent of the events in the past and depend only on the current state.}.

Fig. \ref{fig:statetime} illustrates the time diagram of state changes of such random process which is the state of the forwarding plane in the network being announced to the control plane. 

\begin{figure}[http]
    \center
      \includegraphics[scale=0.55,angle=0]{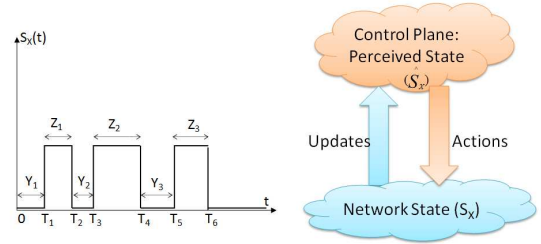}\\
      \caption{\textit{Time diagram of $S_X(t)$: the state of random process $X$ at time $t$.}}
    \label{fig:statetime}
\end{figure}

$\hat{S}_X$ and $S_X$ may differ in two cases; first, when the state of $X$ is changed from $'0'$ to $'1'$ (change type I) but the control plane is not notified ($\hat{S}_X=0,\ S_X=1$); second, when the state of $X$ is changed from $'1'$ to $'0'$ (change type II) and the control plane still has the old information about it ($\hat{S}_X=1,\ S_X=0$). Here we calculate the minimum rate at which the underlying plane has to update the state of $X$ so that the mentioned errors are less than some values $\epsilon_1$ for the first type of error, and $\epsilon_2$ for the second type, respectively. $\epsilon_1$ and $\epsilon_2$ can be viewed as probability of false negative and false positive alarms at the controller.

We make an additional assumption that the delay of the network is negligible with respect to the time scale of the changes in the state of the system, and the control plane will be aware of the announced state immediately (the alternative - that the state of the system changes as fast or faster as the control plane can be notified of these changes - is obviously unmanageable). Thus, the above errors may occur just when the forwarding plane does not send an update about a change.

The main result now can be stated as a Lemma (with the proof in Appendix).

\begin{lemma}\label{lem:01}
	If the ups and downs in the state of $X$ follow some distributions with means $\tau_X$ and $\theta_X$, respectively, then the minimum update rate $R_X(\epsilon_1,\epsilon_2)$ (number of update packets per second) satisfying the mentioned distortion criterion is given by
\begin{eqnarray}
\frac{1}{\tau_X+\theta_X}(2-\frac{\epsilon_1\frac{\theta_X}{\tau_X}}{\frac{\theta_X}{\tau_X+\theta_X}-\epsilon_2}-\frac{\epsilon_2\frac{\tau_X}{\theta_X}}{\frac{\tau_X}{\tau_X+\theta_X}-\epsilon_1}) \label{eq:rX}
\end{eqnarray}
if $\frac{\epsilon_2}{1-\epsilon_2} < \frac{\theta_X}{\tau_X} < \frac{1-\epsilon_1}{\epsilon_1}$ and $\epsilon_2 \tau_X+\epsilon_1 \theta_X < \frac{\tau_X \theta_X}{\tau_X+\theta_X}$. Otherwise an update of rate zero can satisfy the distortion criteria.
\end{lemma}

Equation \ref{eq:rX} shows the minimum update rate for state of a single random variable X in the underlying plane so that an accepted amount of distortion is satisfied. The total rate and consequently the total protocol overhead for keeping the control layer informed of the forwarding layer is the combination of all the overheads needed for all the random processes of the underlying layer, which may be independent of each other or have some impact on each other.
For example, the locations and velocities of different nodes in a mobile wireless network may be considered independent of each other, thus the total overhead will be simply the sum of the individual overheads. But in a caching network, the existence of a content in a specific cache strongly depends on the positions of the other contents if the cache storage sizes are limited.

In the following sections, we apply our model to cache networks and study the total data retrieval cost, including the protocol overhead.

\section{Content Location in ICNs}
\label{sec:scenarios}

Information-centric networks (ICN) require the control layer to know at least one location for each piece of data. However, many ICNs (for instance~\cite{Pursuit}), attempt to set up a route to a nearby copy by requesting the content from a pub/sub mechanism. The pub/sub rendez-vous point needs to know the location of the content. This is highly dynamic, as content can be cached, or expunged from the cache at any time. NDN~\cite{Zhang2010Named} also assumes that the routing plane is aware of multiple locations for a piece of content~\footnote{The routing (in NDN in particular) could know only one route to the content publisher or to an origin server and find cached copies opportunistically on the path to this server. But Fayazbakhsh et al~\cite{Fayazbakhsh2013Less} have demonstrated that the performance of such an ICN architecture would bring little benefit over that of strict edge caching.}.

The request process impacts the cache state, and we make the usual assumption that the items are requested according to a Zipf distribution with parameter $\alpha$; meaning that the popularity of an item $i$ is $\alpha_i=\frac{i^{-\alpha}}{\sum_{k=1}^M k^{-\alpha}}$, where $M$ is the size of the content set, and can be arbitrarily large. We also assume from now on that the Least-Recently-Used (LRU) replacement policy is used in the caches, as it is a common policy and has been suggested in some ICN architectures~\cite{Jacobson2009Networking}. (Other caching policies can be handled in a similar manner.) We denote the cache size by $L_c$. From the popularity $\alpha_i$, Dan and Towsley \cite{Dan1990Approximate}  and Che et al|~\cite{Che2002Hierarchical} provide a model to calculate the probability $\rho_i$ of an item being in the cache under the Independent Reference Model (IRM). Since we know the $\alpha_i$ distribution, we consider that $\rho_i$ is also given.

In the following sections we use the results of section \ref{sec:protocol} together with LRU policy results to study the total control packet rates needed to keep the control plane updated about the items in the caches in two scenarios: 
\begin{itemize}
\item in scenario $\romannumeral 1$, we consider nodes updating the control plane of a domain (say, an AS) so as to route content to a copy of the cache within this domain if its available. We denote the control plane function which locates the content for each request as the Content Resolution System (CRS);
\item in scenario $\romannumeral 2$, we consider two controllers over two neighboring domains updating each other.
\end{itemize}

\subsection{Scenario $\romannumeral 1$: Intra-AS Cache-Controller Interaction}
\label{subsec:scenario1}

We will consider an autonomous network containing $N$ nodes (terminals), each sending requests for items $i=1,...,M$ with sizes $B_i$ according to a Poisson distributed process with rate of $\lambda'_i$. The total request rate for all the items from each node is denoted by $\lambda'=\sum_{i=1}^{M}\lambda'_i$. Note that the total request rate of each terminal is a fixed rate independent of the total number of nodes and items while the total requests for all the nodes is a function of $N$ (namely $N\lambda'$).

Suppose that there are $N_c$ caches in the system ($\mathcal{V}_c=\{v_1,...,v_{N_c}\}$) each with size $L_c$ that can keep (and serve) any item $i$ for some limited amount of time $\tau_i$, which depends on the cache replacement policy. For simplicity, we assume that all the caches are similar to each other~\footnote{We can easily extend to the case of heterogenous caches at the cost of notation complexity. For instance, Theorem~\ref{thm:1} below can be stated as a sum over all $N_c$ possible types of caches with $N_c$ different $\rho_i$'s for each type of cache, instead of a product by $N_c$ of identical terms. Our purpose is to describe the homogenous case, and let the reader adapt the heterogenous case to suit her/his specific needs.} and the rate of requests for item $i$ received by each cache is $\lambda_i=\lambda'_iN/N_c$. Assume that $\bar{N}^c_i$ caches store item $i$ during each download ($\bar{N}^c_i$ could be a single copy near the requester, or multiple copies at different caches).

Whenever a client has a request for an item, it needs to discover a location of that item, preferentially within the AS, and it downloads it from there. To do so, it will ask a (logically) centralized \emph{Content Resolution System (CRS)} or will locate the content by any other non-centralized locating protocol.

If the network domain is equipped with a CRS, it is supposed to have the knowledge of all the caches, meaning that each cache sends its item states (local presence or absence of each item) to the CRS whenever some state changes. 

Depending on the caching policy, whenever a piece of content is being downloaded, either no cache, all the intermediate caches on the path, or just the cache directly connected to the requester stores it in its content store independently of the content state in the other caches, or refresh it if it already contains it.

We want to compute the update rate for this system assuming that each downloaded piece of content is stored only at the cache directly connected to the requester, and demonstrate the following theorem.

\begin{theorem}\label{thm:1}
The total update rate is the summation of the rate $R_{i}(\epsilon_1,\epsilon_2)$ for all $i$ with
\begin{eqnarray}
   R_{i}(\epsilon_1,\epsilon_2)&&\geq N_c\lambda_i (1-\rho_i) \times \label{eq:ri}\\
	&&(2-\frac{\epsilon_1(1-\rho_i)}{\rho_i(1-\rho_i-\epsilon_2)}-\frac{\epsilon_2 \rho_i}{(1-\rho_i)(\rho_i-\epsilon_1)}) \nonumber
\end{eqnarray}

if $\epsilon_1 < \rho_i < 1-\epsilon_2$ and $\epsilon_1 (1-\rho_i)+\epsilon_2 \rho_i < \rho_i(1-\rho_i)$. Otherwise no update is needed.
\end{theorem}

\begin{proof}
Let the random process $X$ in the forwarding plane denote the existence of item $i$ in cache $v_j$ at time $t$, which is needed to be announced to the control plane (CRS). We assume that all the caches have the same characteristics resulting in the same average up and down duration times for each item in all the caches. Let $\tau_i$ denote the mean duration item $i$ spends in any cache $v_j$, and $\theta_i$ denote the mean duration of item $i$ not being in the cache. We define $u_i=\tau_i+\theta_i$.

It can be seen that at the steady-state, the probability of cache $j$ containing item $i$ will be $\rho_i=\frac{\tau_i}{u_i}=\frac{\tau_i}{\theta_i+\tau_i}$.

In order to keep the CRS updated about the content states in the network, all the nodes have to send update packets regarding their changed items to the CRS. All the assumptions of section \ref{sec:protocol} are valid here. Thus, by replacing $\tau_X$ and $\theta_X$ in equation \ref{eq:rX} with $\tau_i$ and $\theta_i$ respectively, the result ($R_{ij}=R_X$) shows the minimum rate at which each cache  $v_j$ has to send information about item $i$ to the CRS.
	
Since we assumed that each cache stores items independent of the items in other caches, the total update rate for item $i$, is the sum of the update rates in all caches which is $R_i(\epsilon_1,\epsilon_2)=N_cR_{ij}(\epsilon_1,\epsilon_2)$, where $R_{ij}$ is the update rate obtained through equation \ref{eq:rX} for item $i$ at cache $v_j$.

The total rate of generating (or refreshing) copies of item $i$ at each cache is $\lambda_i$, which equals to $\frac{1}{\theta_i}$.
Replacing the values of $\frac{\tau_i}{\tau_i+\theta_i}$ and $\frac{1}{\theta_i}$ in $R_{ij}$ with $\rho_i$ and $\lambda_i$ respectively, we can express the total update rate of item $i$ in terms of the probability of this item being in a cache.

This yields the result of equation~\ref{eq:ri} and the total update rate for all the items is the summation of these rates.
\end{proof}

One important consequence of Theorem~\ref{thm:1} is that for large $M$, due to the heavy tail of the popularity distribution, a significant number of the requests will be for items with low $\rho_i$, thus creating an update rate that, according to the first line of equation~\ref{eq:ri}, will not vanish. This means that updating the control plane of an ICN architecture would create a significant amount of traffic with little or not benefits.

Fig. \ref{fig:updatemiss} illustrates $\frac{R_i}{N_c\lambda_i}$ versus $1-\rho_i$. The only parameters that can change this graph are $\epsilon_1$ and $\epsilon_2$. The higher distortion we tolerate, the less update announcements we need to handle. As can be seen the update rate starts from zero for those items which are for sure in the cache, and does not need any CRS update. At the other end of the graph, for the items which are almost surely not in the cache, again no update is needed. The number of items which need some updates is decreasing when higher distortions are accepted.

\begin{figure}[http]
    \center
      \includegraphics[scale=0.18,angle=0]{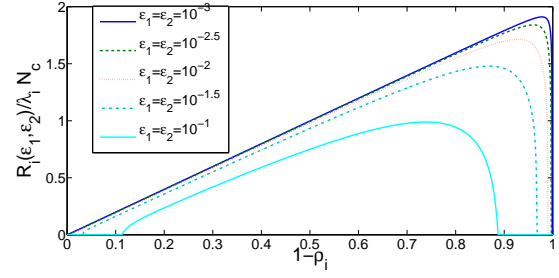}\\
      \caption{\textit{Total item $i$ cache-CRS update rate versus $1-\rho_i$ for different distortion criteria.}}
    \label{fig:updatemiss}
\end{figure}

Fig. \ref{fig:lruLc} shows the changes of the total update rate (scaled by $\frac{1}{\lambda N_c}$) versus the cache storage size, such that the distortion criteria defined by $(\epsilon_1,\epsilon_2)$ is satisfied. In this simulation $M=10^3$. Note that each change in a cache consists of one item entering into and one other item being expunged from the cache, therefore if no distortion is tolerable, this rate will be $2$ updates per change per cache.

\begin{figure}[http]
    \center
      \includegraphics[scale=0.18,angle=0]{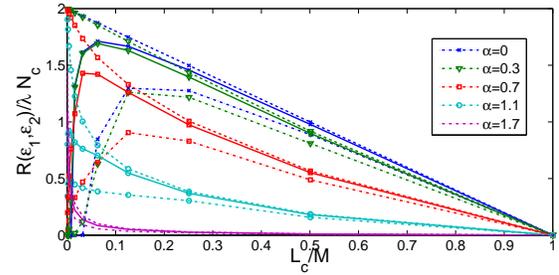}\\
      \caption{\textit{Total cache-CRS update rate (Updates per new arrival per cache) for different cache storage capacities ($M=10^3$) and $\epsilon_1=\epsilon_2=0$(dashed), $.01$(solid), $0.05$(dash-dotted)).}}
    \label{fig:lruLc}
\end{figure}

It can be observed that for very small storage sizes and small popularity index, almost each incoming item will change the status of the cache and a need for an update arises. When the storage size is still very small, the caches do not provide enough space for storing the items and reusing them when needed, so increasing the size will increase the update rate. At some point, the items will move down and up in the cache before going out, so increasing the storage size more than that will reduce the need to update. However, if the popularity index is large, then increasing cache size from the very small sizes will decrease the need to update since there are just a few most popular items which are being requested.

Moreover, as it is expected, the more distortion is tolerable, the CRS needs fewer change notifications. However, if the cache size is too big, or the popularity exponent is too high, fewer changes will occur, but almost all the changes are needed to be announced to the CRS. On the other hand, for small cache sizes accepting a little distortion will significantly decrease the update rate.

To figure out how the calculated rates perform in practice, we simulate an LRU cache with capacity $L_c=20$ items, which are selected from a catalog of size $M=1000$ items and are requested according to Zipf distribution with parameter $\alpha=0.7$. In these simulations we first estimate the item availability in the cache $\rho_i$, then using these estimated $\rho_i$ and according to equations \ref{eq:U1} and \ref{eq:U2}, we calculate the update probability in case of a change. We then run the simulation for $50000$ Poisson requests and measure the average generated distortion during 20 rounds of simulation. Fig. \ref{fig:EvalDist1} illustrates the results for the case where $\epsilon_1=\epsilon_2=0.01$.

\begin{figure}[http]
    \center
      \includegraphics[scale=0.18,angle=0]{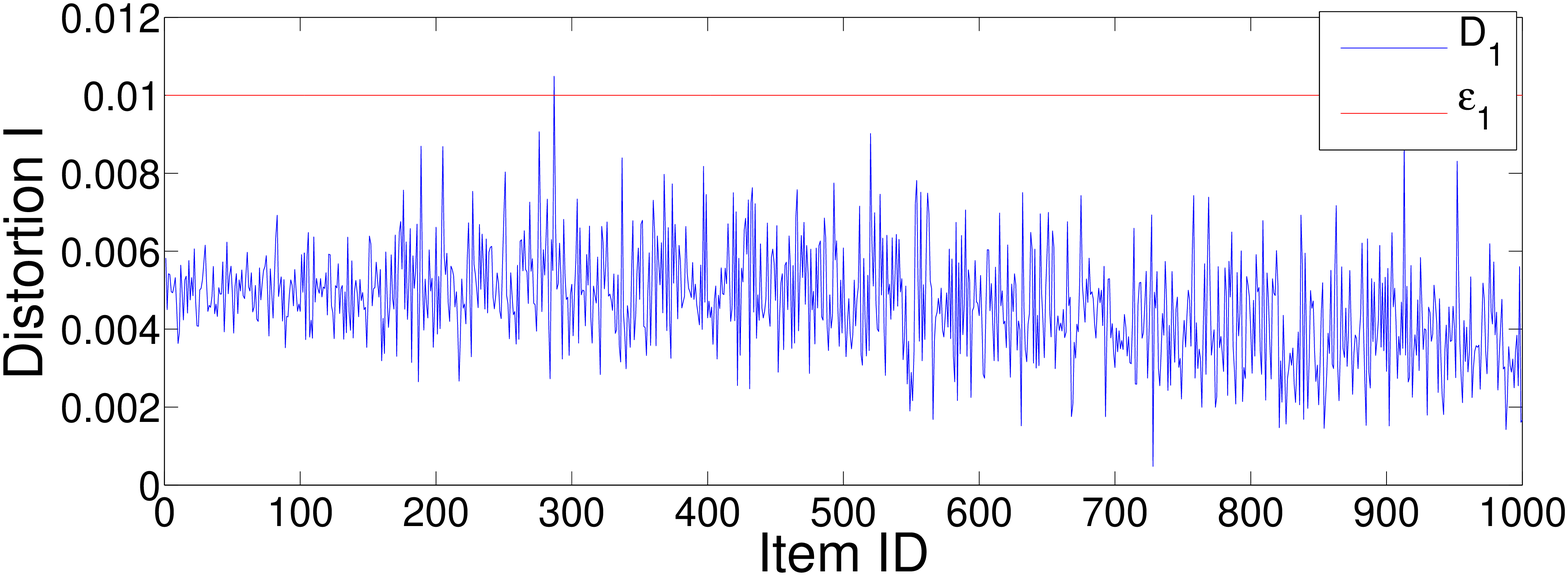}\\
			\includegraphics[scale=0.18,angle=0]{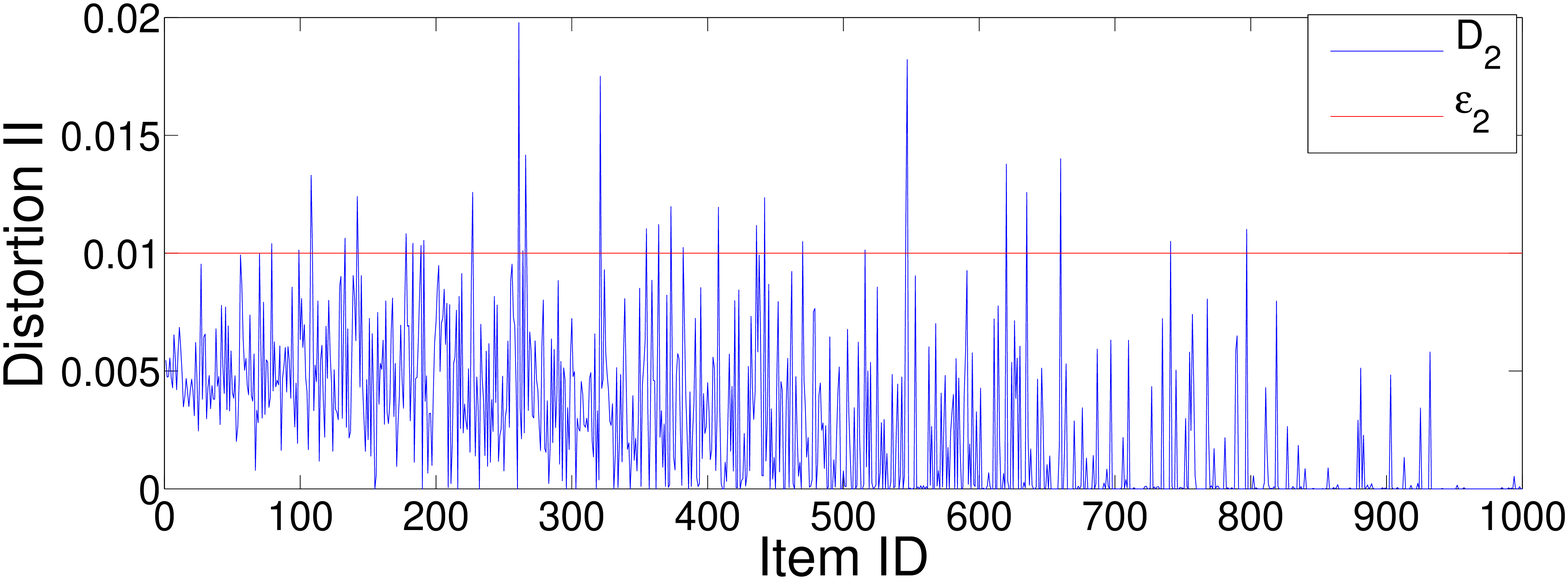}\\
      \caption{\textit{Measured distortion type I ($D_1$) and II ($D_2$) for $\epsilon_1=\epsilon_2=0.01$, $M=1000$, $L_c=20$, and $\alpha=0.7$.}}
    \label{fig:EvalDist1}
\end{figure}

It is observed that for a large portion of the items the distortion type I satisfies the distortion criteria. Distortion type II, however, has more unsatisfied distortion. The reason is that the calculated update rates are strongly dependent on the availability of the items in the cache and any small error in the estimation of $\rho_i$ may lead to some extra distortions. Since the $\rho_i$'s are mostly very small, not updating just one type II change may cause an error which remain in the system for a long time, and thus creating a large distortion.


\subsection{Scenario $\romannumeral 2$: Interaction Between Two ASs}
\label{subsec:scenario2}

In this section, we assume two separate neighboring Autonomous Systems $(AS_1$, $AS_2$) similar to the one discussed in section \ref{subsec:scenario1}. Each AS can be considered as a big storage containing all the items of its caches. Let $N^{(j)}_C$ and $L^{(j)}_C$ denote the number and size of the caches inside $AS_j$, and $P^{(j)}_i$ denote the probability of item $i$ being in at least one cache in $AS_j$ for $j=1,2$.

The control plane in each AS is informed about all the contents of the caches inside that AS through the mechanism described earlier in previous sections. The control planes (CRSs) of two separate ASs also need to be informed of the information stored in the other AS so that they can forward the requests to the proper AS in case of not being able to serve that request using the local caches.

Here we try to derive the minimum rate at which one AS ($AS_1$) needs to send information about its stored items to the other AS ($AS_2$) such that some distortion criteria is satisfied. Similar to the case of information updates inside one AS, two types of errors may happen; first, when $AS_1$ contains item $i$ and $AS_2$ is not aware of that; second, when $AS_1$ does not contain item $i$ and $AS_2$ assumes otherwise. Since each control plane first tries to find a copy of the requested data in some caches inside its own sub-network, none of these errors are  important if $AS_2$ itself contains item $i$. Therefore, the distortion happens only when item $i$ is not stored in any cache inside $AS_2$, and the distortion criteria is defined based on these errors.
$Pr(S_{i_2}=0, S_{i_1}=1, \hat{S}_{i_1}=0)<\epsilon_1$ and $Pr(S_{i_2}=0, S_{i_1}=0, \hat{S}_{i_1}=1)<\epsilon_2$.

$S_{i_j}$ is the state of item $i$ in $AS_j$ , and $\hat{S}_{i_j}$ denotes the state of item $i$ in $AS_j$ perceived by the other AS ($j=1,2$).

We assume that the existence of item $i$ in different ASs are independent of each other. Consequently, the perceived state of item $i$ in one AS by the other AS is independent of the existence of that item in the latter AS. Therefore,

\begin{equation}
Pr(S_{i_2}=0)Pr(S_{i_1}=1, \hat{S}_{i_1}=0)<\epsilon_1 \nonumber
\end{equation}
\begin{equation}
Pr(S_{i_2}=0)Pr(S_{i_1}=0, \hat{S}_{i_1}=1)<\epsilon_2
\end{equation}

Replacing $Pr(S_{i_2}=0)$ by $1-P^{(2)}_i$ we will have
\begin{eqnarray}
Pr(S_{i_1}=1, \hat{S}_{i_1}=0)<\frac{\epsilon_1}{1-P^{(2)}_i} \nonumber \\
Pr(S_{i_1}=0, \hat{S}_{i_1}=1)<\frac{\epsilon_2}{1-P^{(2)}_i}
\end{eqnarray}

$(1-P^{(2)}_i)$ is the probability of item $i$ not being available in $AS_2$ and is equal to $(1-\rho^{(2)}_i)^{N^{(2)}_C}$, where $\rho^{(2)}_i$ denotes the probability of item $i$ in each cache of $AS_2$. Let $\epsilon'_1=\frac{\epsilon_1}{1-P^{(2)}_i}$ and $\epsilon'_2=\frac{\epsilon_2}{1-P^{(2)}_i}$.
According to Lemma \ref{lem:01} the total information rate regarding item $i$ in $AS_1$ reported to $AS_2$ is given by
\begin{eqnarray}
	R^{(1)}_i(\epsilon'_1,\epsilon'_2)&\geq& \frac{1}{\tau^{(1)}_i+\theta^{(1)}_i} \times  \label{eq:r1i}  \\
&&(2-\frac{\epsilon'_1\frac{\theta^{(1)}_i}{\tau^{(1)}_i}}{\frac{\theta^{(1)}_i}{\tau^{(1)}_i+\theta^{(1)}_i}-\epsilon'_2}-\frac{\epsilon'_2\frac{\tau^{(1)}_i}{\theta^{(1)}_i}}{\frac{\tau^{(1)}_i}{\tau^{(1)}_i+\theta^{(1)}_i}-\epsilon'_1}) \nonumber
\end{eqnarray}

if $\frac{\epsilon'_2}{1-\epsilon'_2} < \frac{\theta^{(1)}_i}{\tau^{(1)}_i} < \frac{1-\epsilon'_1}{\epsilon'_1}$ and $\epsilon'_2 \tau^{(1)}_i+\epsilon'_1 \theta^{(1)}_i < \frac{\tau^{(1)}_i \theta^{(1)}_i}{\tau^{(1)}_i+\theta^{(1)}_i}$. Otherwise an update rate of zero can satisfy the distortion criteria.
	
In the above equations $\tau^{(1)}_i$ and $\theta^{(1)}_i$ are the average durations where item $i$ is available in $AS_1$ and the durations where it is not, respectively.

Assume that the rate of generating or refreshing item $i$ in at least one cache of $AS_1$ is denoted by $\lambda^{(1)}_i$. This rate is equal to the total rate of requests for item $i$ from all the users connected to $AS_1$ which equals to $N\lambda_i$. Similar to the reasoning in the proof of Theorem \ref{thm:1} we can derive the results which can now be stated as the following Theorem:

\begin{theorem}\label{thm:2}
The total information rate $R^{(1)}(\epsilon_1,\epsilon_2)$ can be written as:
\begin{eqnarray}
   &R^{(1)}(\epsilon_1,\epsilon_2)\geq \sum_{i=1}^M R^{(1)}_i(\epsilon_1,\epsilon_2) \mbox{ where}&\nonumber\\
   &R^{(1)}_i(\epsilon_1,\epsilon_2)\geq N\lambda_i (1-P^{(1)}_i)\left(2 - \right.& \nonumber  \\
   &\frac{\epsilon_1(1-P^{(1)}_i)}{P^{(1)}_i(1-P^{(1)}_i)(1-P^{(2)}_i)-\epsilon_2P^{(1)}_i}-& \nonumber\\
	 &\frac{\epsilon_2P^{(1)}_i}{P^{(1)}_i(1-P^{(1)}_i)(1-P^{(2)}_i)-\epsilon_1(1-P^{(1)}_i)})& \nonumber\\
   &\mbox{for } \frac{\epsilon_1}{1-P^{(2)}_i} < P^{(1)}_i < 1-\frac{\epsilon_2}{1-P^{(2)}_i} \mbox{ and}&\nonumber \\
   &\epsilon_2P^{(1)}_i+\epsilon_1(1-P^{(1)}_i) < P^{(1)}_i(1-P^{(1)}_i)(1-P^{(2)}_i)& \label{eq:cond1}\nonumber\\
   &\mbox{and }R^{(1)}_i(\epsilon_1,\epsilon_2)=0 \mbox{ otherwise.}&
\end{eqnarray}

In these equations $P^{(1)}_i=1-(1-\rho^{(1)}_i)^{N^{(1)}_C},\ P^{(2)}_i=1-(1-\rho^{(2)}_i)^{N^{(2)}_C}$
and $\rho^{(1)}_i$ and $\rho^{(2)}_i$ can be calculated as in\cite{Dan1990Approximate}.
\end{theorem}

According to the conditions where nonzero update rates are required in the statement of Theorem~\ref{thm:2}, it can be observed that if an item is in a cache in an AS with high probability (large $P^{(2)}_i$), very few  control packets are needed to inform it of the same item in another AS. In other words, when an AS contains an item with high probability, it does not need to know about that item's status in the other AS and the distortion criteria is satisfied with very low rate or even no updates.

\begin{figure}[http]
    \center
      \includegraphics[scale=0.18,angle=0]{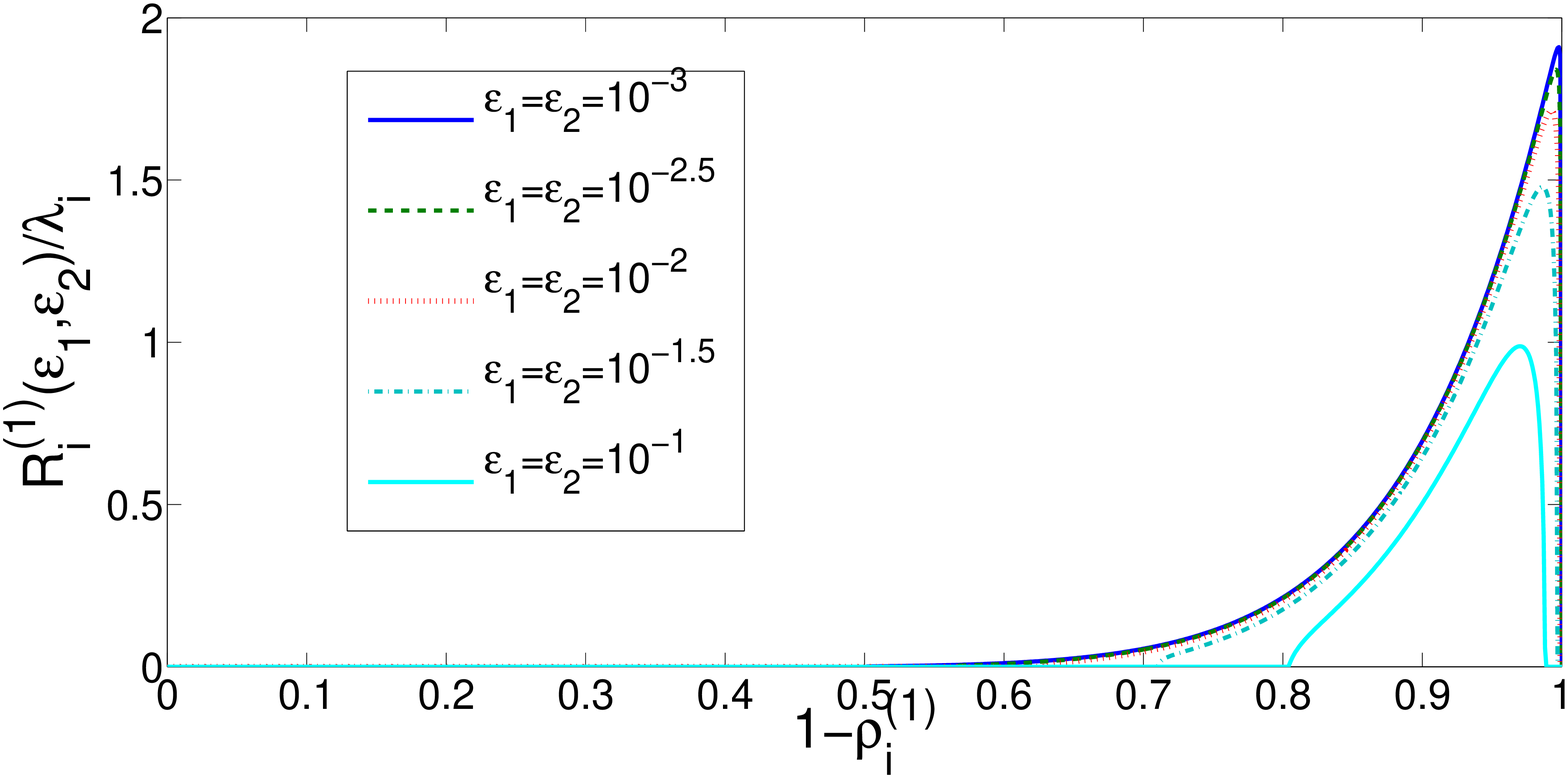}\\
      \caption{\textit{Total item $i$ AS-AS update rate versus $1-\rho_i$ for different distortion criteria. (First AS contains $N^{(1)}_c=10$ caches, and item $i$ is in the second AS with low probability ($P^{(2)}_i=0.001$))}}
    \label{fig:updatemiss2AS}
\end{figure}

In Fig. \ref{fig:updatemiss2AS} we fix the probability of $i$ being in $AS_2$ to a low value ($0.001$) and plot the changes of the ratio of this rate to the request rate versus the probability of item $i$ not being in each cache in $AS_1$, which we assume contains $N^{(1)}_c=10$ caches.
If the item is not in a cache in an AS with high probability, but the probability of that item being in another cache is high  (low $1-\rho_i$ resulting in high $P^{(1)}_i$), very few updates will be needed. Increasing $1-\rho^{(1)}_i$ will decrease $P^{(1)}_i$, so a higher update rate is needed. When this probability is higher than some point, then $P^{(1)}_i$ is very low and the probability of any change in this item is very low, thus again very few updates may be needed.

\section{Application to Cost Analysis}
\label{sec:costanalysis}


Fig. \ref{fig:netmodel} illustrates the network model studied in this section. This model consists of entities in three substrates: users are located on the first layer; a network of caches with the CRS on the second level; external resources (caches in other networks, Internet, etc.) on the third.

\begin{figure}[http]
    \center
      \includegraphics[scale=0.4,angle=0]{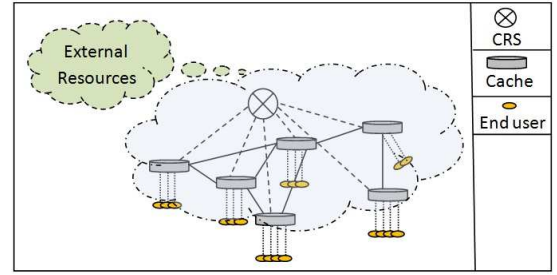}\\
      \caption{\textit{Network Model.}}
    \label{fig:netmodel}
\end{figure}

In this section, we try to calculate the total cost (download cost+update cost) in scenario $\romannumeral 1$ and look at the trade-offs between the cost, the number of caches, and the size of caches. We need to define the relative costs of the different actions. We assume that the state update process for item $i$ has a per bit cost of $\xi^{up}_i$ for sending data from the cache to CRS. 
On the other hand, the requested piece of content $i$ may be downloaded from the local cache with cost 0 (with probability $\rho_i$ of being in this cache), from another cache inside the same network with some per bit cost $\xi^{int}_i$ (with a probability we denote by $P_i - \rho_i$, where $P_i$ is the probability that content $i$ is within the AS's domain), or it must be downloaded from an external server with some other cost $\xi^{ext}_i>\xi^{int}_i$ (with probability $(1-P_i)$). Obviously, $\rho_i \leq P_i \leq 1$. 

%

The total download cost of item $i$ with size $B_i$ bits in the sub-network is
\begin{equation}
\varphi^{dl}_i = \lambda_i \times B_i \times ((P_i-\rho_i) \times \xi^{int}_i + (1-P_i) \times \xi^{ext}_i),
\end{equation}

Each cache sends update packets to provide its CRS with the state of item $i$ in its local content store. Each update packet contains the ID of the cache issuing the query, the ID of the updated item and the new state. $\log N_c$ bits are needed to represent the cache. Item $i$ is updated with probability $\beta_i=\frac{\lambda_i(1-\rho_i)}{\sum_{k=1}^M\lambda_k(1-\rho_k)}$, which results in a code length of at least $-\log \beta_i$ bits. Thus, the length of each update packet is $l_i\geq \log N_c-\log \beta_i +1$. Hence, the total cost for updating information about item $i$ in the sub-network is $\varphi^{up}_i = R_i(\epsilon_1,\epsilon_2) \times l_i \times \xi^{up}_i$, where $R_i(\epsilon_1,\epsilon_2)$ is the minimum rate at which the update state of item $i$ must be announced to CRS so that a distortion criteria defined by $(\epsilon_1,\epsilon_2)$ is satisfied.

The total cost for item $i$ is the sum of the update and download costs: $\varphi_i =\varphi^{dl}_i + \varphi^{up}_i$. Therefore,$\varphi_i$ is
\begin{eqnarray}
\lambda_i B_i  ((P_i-\rho_i) \xi^{int}_i + (1-P_i) \xi^{ext}_i) + R_i(\epsilon_1,\epsilon_2) l_i \xi^{up}_i
\end{eqnarray}

Then the total cost for all the items is
\begin{eqnarray}
\varphi &=& \sum_{i=1}^M\left[\lambda_i B_i \left((P_i-\rho_i)\xi^{int}_i + (1-P_i)\xi^{ext}_i\right) \right.\nonumber\\ 
& & \left. + R_i(\epsilon_1,\epsilon_2)l_i \xi^{up}_i \right] \label{eq:totcost}
\end{eqnarray}
We now compute some bounds on $P_i$ based upon the allowed distortion.

Recall that  $\mathcal{V}_c$ is the set of caches, $\rho_i$ denotes the probability that a specific cache contains item $i$, $S_{ic}$ represents the state of an item $i$ at a node $j$, which is $1$ if cache $c$ contains item $i$, and $0$ otherwise, and $\hat{S}_{ic}$ denotes the corresponding state perceived by the CRS. A request from a user is not served internally (by a cache in second layer) either if no cache contains it \begin{equation} Pr(\forall\ c\in \mathcal{V}_c:S_{ic}=0)=(1-\rho_i)^{N_c}, \end{equation}
or if there are some caches containing it but the CRS is not aware of that.
\begin{eqnarray}
  &&Pr(\exists\ c \in \mathcal{V}_c:\ i \in c\ \&\ \hat{S}_{ic}=0) \nonumber \\
  &&=\sum_{k=1}^{N_c} \sum_{1\leq c_1<..<c_k\leq N_c} Pr(^{i \notin \mathcal{V}_c-\{c_1,...,c_k\}\ \&\ }_{[\hat{S}_{ic_l}=0,\ S_{ic_l}=1]^k_{l=1}}) \nonumber \\
  &&=\sum_{k=1}^{N_c} \sum_{1\leq c_1<..<c_k\leq N_c} (1-\rho_i)^{N_c-k} \Pi_{l=1}^k Pr(^{\hat{S}_{ic_l}=0}_{S_{ic_l}=1}) \nonumber \\
  &&=\sum_{k=1}^{N_c} (^{N_c}_k) (1-\rho_i)^{N_c-k} \sigma_{1_i}^k  \nonumber \\
	&&=(1-\rho_i+\sigma_{1_i})^{N_c}-(1-\rho_i)^{N_c} \nonumber
\end{eqnarray}
where $\sigma_{1_i}\geq 0$ is the probability that $i$ exists in cache $c$ and the CRS does not know about it.

Thus the probability that a request is served externally is $1-P_i$ which equals
\begin{eqnarray}
&(1-\rho_i)^{N_c} +((1-\rho_i+\sigma_{1_i})^{N_c}-(1-\rho_i)^{N_c})& \nonumber \\
&= (1-\rho_i+\sigma_{1_i})^{N_c},& \label{eq:Pext}
\end{eqnarray}
where under the independent cache assumption, the state of an item in a cache is independent of the state in another cache. The probability $\sigma_{1_i}\geq 0$ is always less than the probability of $i$ being in cache $j$ ($\sigma_{1_i}\leq \rho_i$). If the state updates are done at rate greater than $R_i(\epsilon_1,\epsilon_2)$, it will be less than $\epsilon_1$. Hence, denoting $[x]^+ = \max(x,0)$, 
\begin{equation}
[1-(1-\rho_i+\epsilon_1)^{N_c}]^+\leq P_i \leq 1-(1-\rho_i)^{N_c}.
\end{equation}


Fig. \ref{fig:lruFixLcNc} illustrates the changes of update and total cost when the size of each cache is limited to $L_c=100$. The request rate received by each cache is inversely proportional $N_c$ (the request rate per user is assumed to be fixed and independent of $N_c$), and the update packet length increases logarithmically with the number of caches. The total update rate per cache is almost linearly decreasing with $N_c$, hence the total update rate will almost be stable when $N_c$ varies (changes are in the order of $\log N_c$). Increasing $N_c$, however, increases the probability of an item being served internally and thus decreases the download, and consequently the total cost.
\begin{figure}[http]
    \center
      \includegraphics[scale=0.18,angle=0]{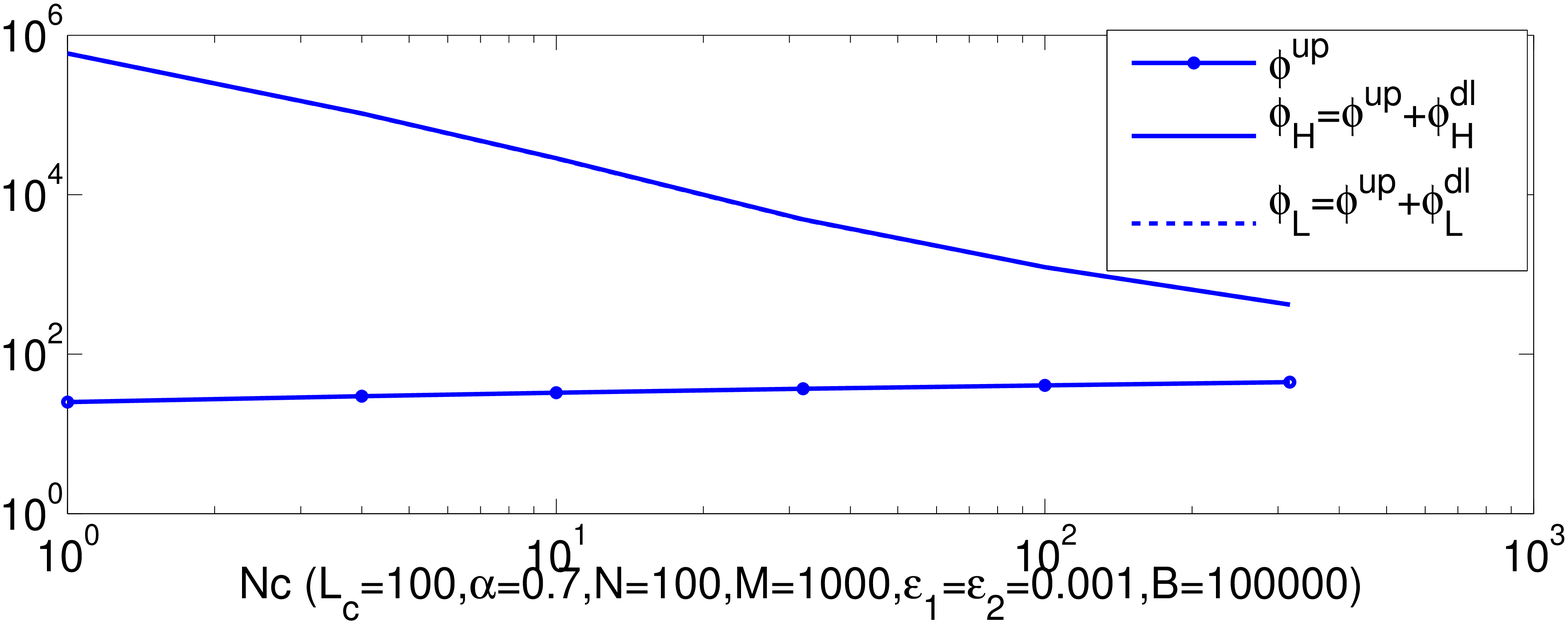}\\
      \caption{\textit{Scenario $\romannumeral 1$, Total update cost ($\varphi^{up}$) and Total Cost (lower $\varphi_L$ and upper bounds $\varphi_H$), when the storage size per cache is fixed ($L_c=100$), vs. the number of caches ($N_c$).}}
        \label{fig:lruFixLcNc}
      \includegraphics[scale=0.18,angle=0]{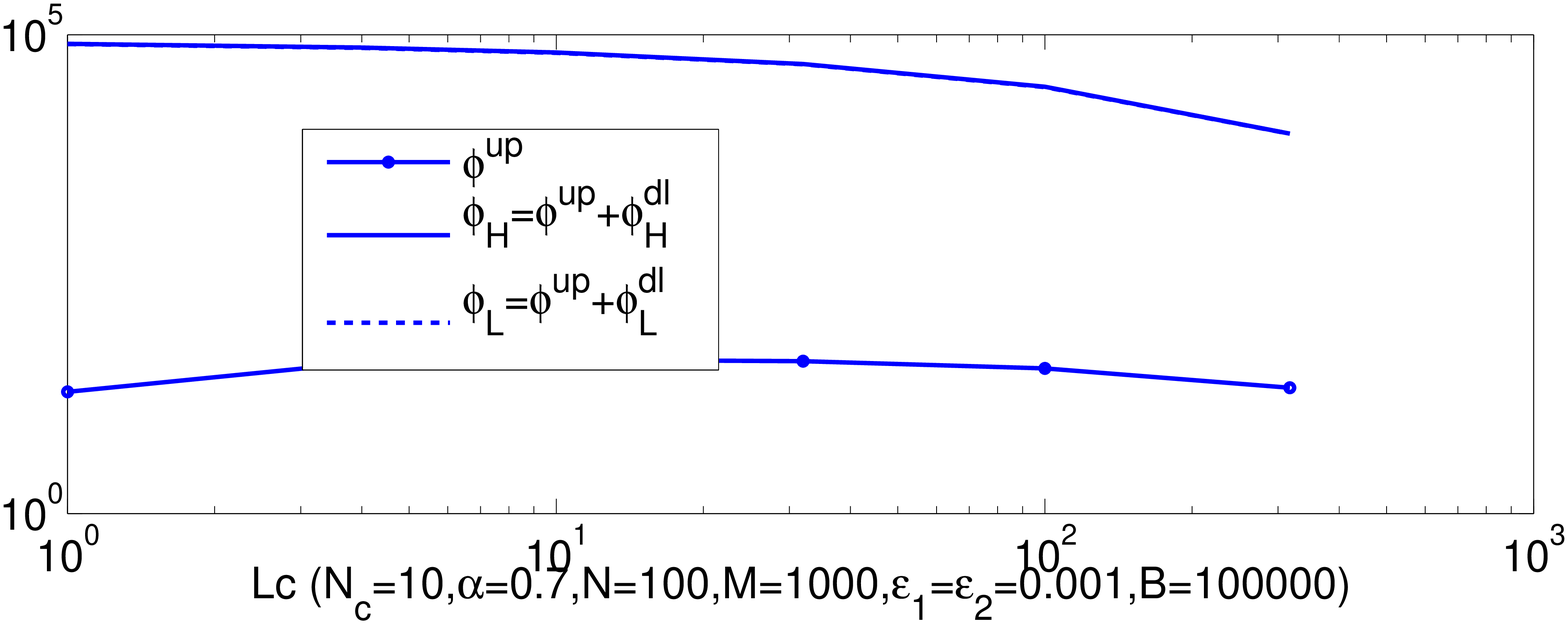}\\
      \caption{\textit{Scenario $\romannumeral 1$, Total update cost ($\varphi^{up}$) and Total Cost (lower $\varphi_L$ and upper bounds $\varphi_H$),when the number of caches is fixed ($N_c=10$), vs. the size of caches ($L_c$).}}
    \label{fig:lruFixNcLc}
\end{figure}

In Fig. \ref{fig:lruFixNcLc}, we fix the number of caches in the AS ($N_c=10$) and study the effects of cache storage size on the update and total cost. Increasing the cache size simply increases the probability of an item being served internally and decreases the download cost. However, the update cost shows more complicated behavior when changing the storage size. Looking at each cache, very small cache size leads to very large durations where that item is not in that cache and consequently, the update rate would be low. Increasing the storage size will increase the probability of that item being in the cache, and thus increases the update rate. This increase will reach its highest value for a certain value of cache size.
For larger values of cache beyond a threshold, the item is in the cache most of the time. Therefore, we need less updates and increasing the cache size will increase the duration of the item being in the cache leading to lower update messages. Since the total cost mostly depends on the download cost, by increasing the cache size, this value reaches its minimum value.

\subsection{Optimized Cache Management}
\label{subsec:optimumi}

We now turn our attention to minimizing the total cost for given $N_c$ and $L_c$.

 Under a Zipf popularity distribution, many rare items will not be requested again while they are in the cache under the LRU policy. We can rewrite the total cost if the caches only keep the items with popularity from 1 up to $i^*$. 
\begin{eqnarray}
\varphi &=& \sum_{i=1}^{i^*} \lambda_i  B_i \left((P_i-\rho_i)\xi^{int}_i + (1-P_i)\xi^{ext}_i\right)\nonumber \\ &+& \sum_{i=i^*+1}^M  \lambda_i  B_i \xi^{ext}_i + \sum_{i=1}^{i^*} R_i(\epsilon_1,\epsilon_2) l_i \xi^{up}_i  \label{eq:phiwthr}
\end{eqnarray}


Now just $i^*$ different pieces of content may be stored in each cache, so probability of an item $i=1,...,i^*$ being in a cache ($\rho_i$) is changed, which in turn changes $P_i$ and $R_i$.

Fig. \ref{fig:optIstar} demonstrates the total cost versus the caching popularity threshold $i^*$, for different number and size of content stores, and acceptable distortions.

If just a very small number of items (small $i^*$) are kept inside cache layer, then the download cost for those which are not allowed to be inside caches will be the dominant factor in the total cost and will increase it. On the other hand, if a lot of popularity classes are allowed to be kept internally, then the update rate is increased and also the probability of the most popular items being served internally decreases, so the total cost will increase. There is some optimum caching popularity threshold where the total cost is minimized. This optimum threshold is a function of the number and size of the stores, distortion criteria, per bit cost of downloads and updates.

\begin{figure}[http]
    \center
      \includegraphics[scale=0.18,angle=0]{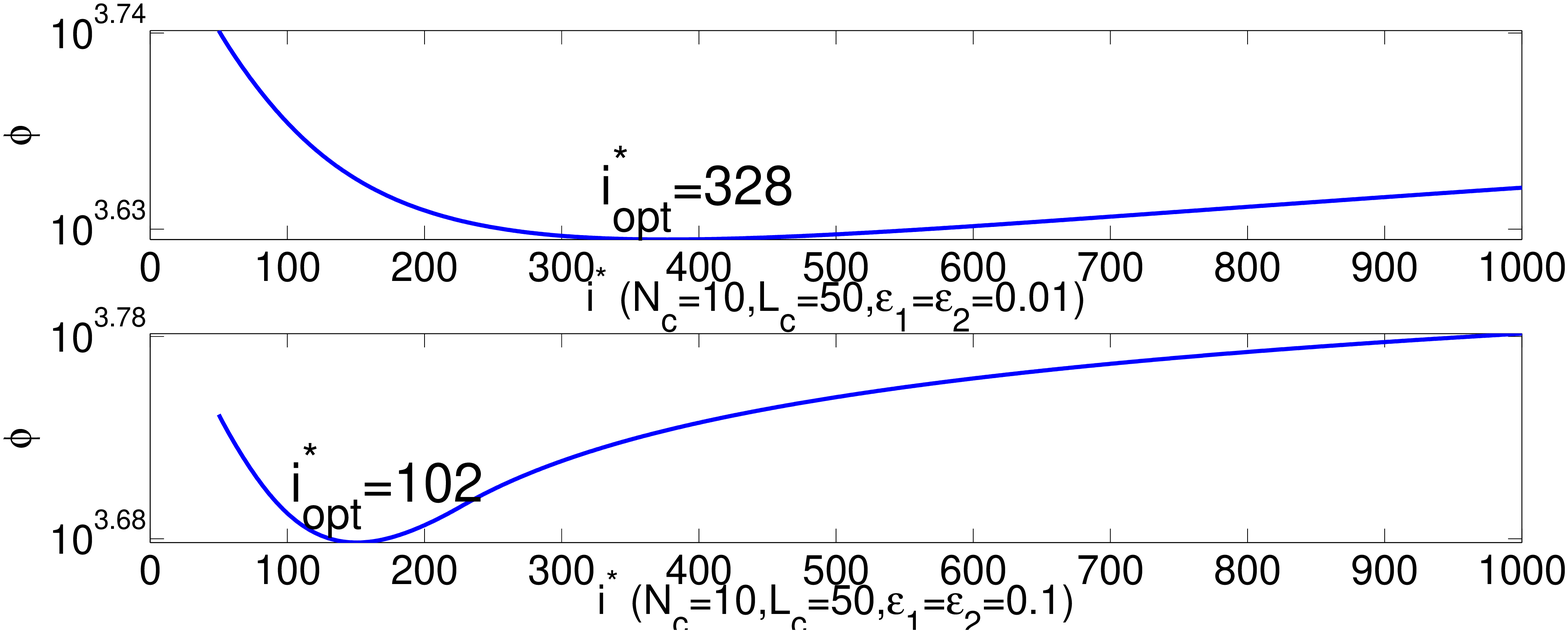}\\
      \includegraphics[scale=0.18,angle=0]{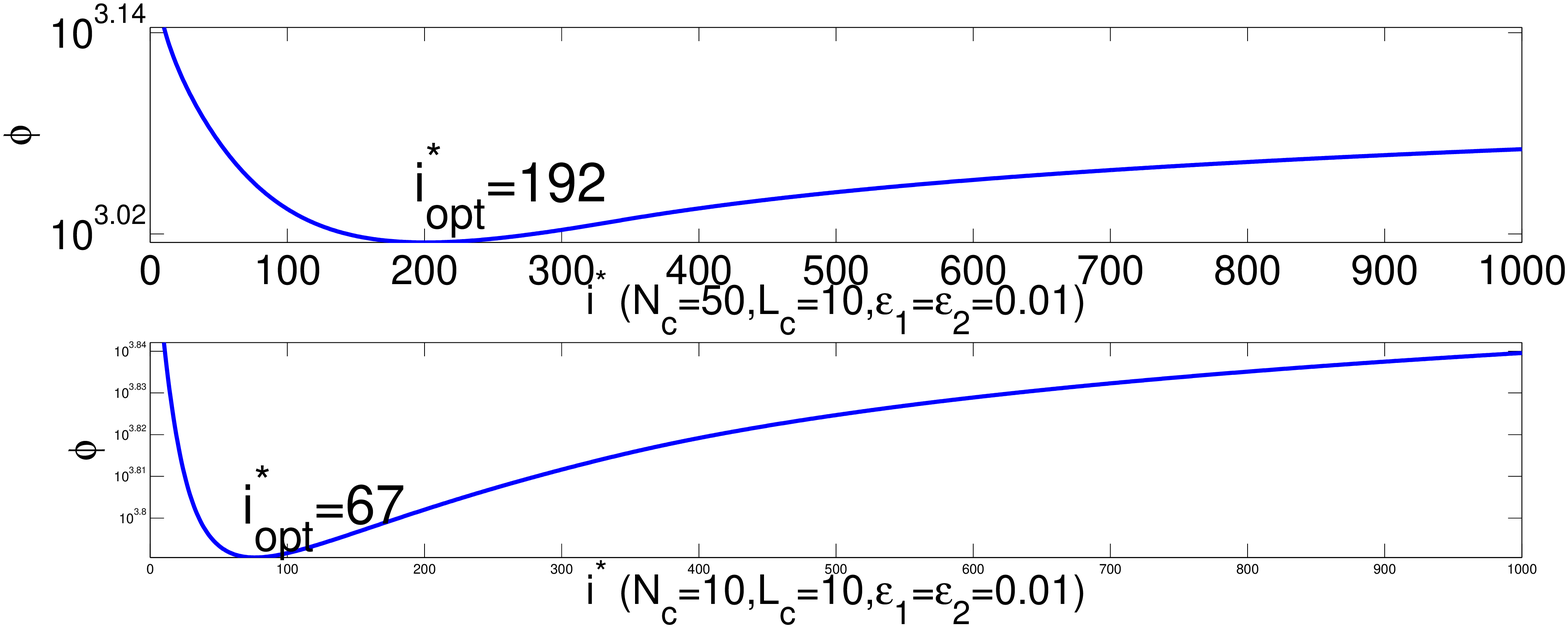}\\
      \caption{\textit{Update, Download, and Total cost when only the $i^*$ most popular items are allowed to be stored inside caches ($N=100,M=1000,B=10^4$).}}
    \label{fig:optIstar}
\end{figure}

To find the optimal $i^*$, assume that all the items have the same size ($B_i=B$) and the per bit costs is fixed for all popularity classes ($\xi_i^{int}=\xi^{int}$, $\xi_i^{ext}=\xi^{ext}$, $\xi_i^{up}=\xi^{up}$). We can rearrange equation \ref{eq:phiwthr}:
\begin{eqnarray}
\label{eq:phiministar}
\varphi &=& \varphi_1 - \varphi_2 + \varphi_3
\end{eqnarray}
where $\varphi_1=B \lambda \xi^{ext}$ is the total cost if no  cache  exists and all the  requests are served externally; $\varphi_2=B\lambda(\xi^{ext}-\xi^{int})\sum_{i=1}^{i^*} \alpha_i P_i + B\lambda \xi^{int} \sum_{i=1}^{i^*} \alpha_i \rho_i$) corresponds to the benefit of caching (cost reduction due to caching); and $\varphi_3=\xi^{up} \sum_{i=1}^{i^*} R_i(\epsilon_1,\epsilon_2)l_i$ is the caching overhead cost due to the updates.  We need to calculate the value of $i^*$ such that the cost of caching is dominated by its advantage; i.e. we need to maximize $\varphi_2-\varphi_3$.

This can be done using numerical methods which will lead to a unique $i^*$ for each network setup (fixed parameters). However, the network characteristics and the request pattern are changing over time, so it seems that it is better to have a mechanism to dynamically optimize the cost by selecting the caching threshold ($i^*$) according to the varying network features.



In such a mechanism, the CRS can keep track of requests and have an estimation of their popularity. For those requests which are served locally the CRS can have an idea of the popularity based on the updates that receives from all the caches; i.e. the longer an item stays in a cache, the more popular it is. It can also take into account the local popularity of the items. The CRS can then dynamically search for the caching threshold which minimizes the total cost by solving equation~\ref{eq:phiministar}. Once the CRS determines which items to keep internally, it will set/reset a flag in each CRRep so that the local cache knows to store or not to store the requested piece of content.

\subsection{Total Cost in the AS-AS Scenario $\romannumeral 2$}

The AS-AS interaction to keep each other updated will add an overhead to the total cost calculated in equation \ref{eq:totcost}. So assuming $\xi^{extup}_i$ for the per bit cost of the control packets between the ASs, the total cost is
\begin{eqnarray}
&\varphi =\sum_{i=1}^M \left[ \lambda_i B_i \left((P_i-\rho_i) \xi^{int}_i \right. \right.& \nonumber\\ 
&+ \left. (1-P_i) \xi^{ext}_i\right) + R_i(\epsilon_1,\epsilon_2) l_i \xi^{up}_i + \left. R^{(1)}_i l^{(1)}_i \xi^{extup}_i \right]&
\end{eqnarray}
where $l^{(1)}_i$ is the length of the control packet sent from $AS_1$ to $AS_2$. Let $\beta^{(1)}_i$ denote the probability of a change in item $i$'s status in $AS_1$. This probability equals the probability of an item $i$ change in an $AS_1$ cache while no other caches contains it.
\begin{eqnarray}
\beta^{(1)}_i = \frac{(1-\rho_i)^{N_c-1}\beta_i}{\sum_{j=1}^M (1-\rho_j)^{N_c-1}\beta_j} \nonumber
\end{eqnarray}
Assuming $N_{AS}$ number of ASs, the inter-AS control packet length for updating item $i$'s status is $l^{(1)}_i = \log N_{AS}-\log \beta^{(1)}_i +1$.

Fig. \ref{fig:cost2ASvsLc} shows the internal update cost (cache-CRS update cost), external update cost (AS-AS update cost), and total cost versus the storage size for fixed number of caches ($N_C=10$). In this Figure we assume $\xi^{extup}_i=\xi^{ext}_i$, and $\xi^{up}_i=\xi^{int}_i$. Increasing the storage size of each cache increases the probability of an item being in at least a cache in an AS and decreases the rate of updates and the external update cost. For low storage size, the AS-AS update rate will be almost the same as cache-CRS update rate, thus the external update cost will be higher than the internal update cost by a factor proportional to $\xi^{extup}_i/\xi^{up}_i$ and $E[l^{(1)}_i]/E[l_i]$. For large enough storage size, the probability of an item being in at least one cache is much more than the probability of it being in a specific cache, so the AS-AS update rate will be lower than the cache-CRS update rate and the corresponding costs are being closer to each other.

\begin{figure}[http]
    \center
      \includegraphics[scale=0.18,angle=0]{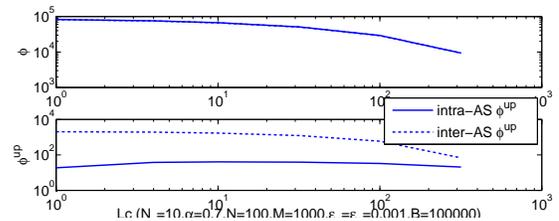}\\
      \caption{\textit{Scenario $\romannumeral 2$, (a) Total Cost, (b) Total update cost (intra-AS and inter-AS) vs. the size of caches (Lc).}}
    \label{fig:cost2ASvsLc}
\end{figure}



\section{Conclusions}
\label{sec:conclusion}

We formulated a distortion-based protocol overhead model for ICNs. Some simple content distribution networks were then considered as examples to show how this framework can be used. We calculated the overhead of keeping the control plane informed about the state of the content. We also studied the total cost of data retrieval and observe that with limited cache storage sizes, allowing all the items to have the opportunity to be stored inside the network's caches is not efficient. For the case with a central resolution system in each sub-network and with LRU cache replacement policy, an algorithm has been proposed that can dynamically determine which items should be cached inside the AS at any time such that the total cost of data retrieval is minimized.

\bibliographystyle{abbrv}
\bibliography{acmicn14}

\appendix
{\textbf{Proof of Lemma~\ref{lem:01}: }}The distortion criteria is defined as $D_1 = Pr(S_X=1,\hat{S}_X=0) \leq \epsilon_1$ and $D_2 = Pr(S_X=0,\hat{S}_X=1) \leq \epsilon_2$. It can be seen that
\begin{equation}
Pr(S_X=1)=\frac{\tau_X}{\tau_X+\theta_X} \mbox{ and }Pr(S_X=0)=\frac{\theta_X}{\tau_X+\theta_X} \mbox{.   }\nonumber
\end{equation}
There are three cases where a zero update rate can satisfy the distortion criteria.\\
$\bullet$ 1) If $Pr(S_X=1)\leq \epsilon_1$, then keeping $\hat{S}_X$ constantly equal to $'0'$ will result in $D_1=Pr(S_X=1)\leq \epsilon_1$ and $D_2=0<\epsilon_2$.\\
$\bullet$ 2) If $Pr(S_X=0)\leq \epsilon_2$, then keeping $\hat{S}_X$ constantly equal to $'1'$ will result in $D_1=0< \epsilon_1$ and $D_2=Pr(S_X=0)\leq \epsilon_2$.\\
$\bullet$ 3) If $\epsilon_1 Pr(S_X=0)+\epsilon_2 Pr(S_X=1) \geq Pr(S_X=1)Pr(S_X=0) $, then we can find some probability $1-\frac{\epsilon_1}{Pr(S_X=1)} \leq \rho_0 \leq \frac{\epsilon_2}{Pr(S_X=0)}$, such that assigning '1' to $\hat{S_X}$ randomly  and independently of $S_X$ results in $D_1=(1-\rho_0)Pr(S_X=1)\leq \epsilon_1$, and $D_2=\rho_0Pr(S_X=0) \leq \epsilon_2$.

Thus in the following we concentrate on the cases where $Pr(S_X=1)>\epsilon_1$, $Pr(S_X=0)>\epsilon_2$, and $\epsilon_1 Pr(S_X=0)+\epsilon_2 Pr(S_X=1) < Pr(S_X=1)Pr(S_X=0) $. Note that we assume that $\epsilon_1+\epsilon_2\leq 1$, then $\frac{\epsilon_2}{1-\epsilon_2}\leq \frac{1-\epsilon_1}{\epsilon_1}$, and the first two regions can be summarized in the region where $\frac{\epsilon_2}{1-\epsilon_2}\leq \frac{\theta}{\tau} \leq \frac{1-\epsilon_1}{\epsilon_1}$.

Let $U^{1}_X(\epsilon_1)$ and  $U^{2}_X(\epsilon_2)$ denote the ratio of updated type I and II changes to the total number of corresponding changes, respectively, such that the distortion criteria is satisfied. The false negative alarm is generated during the $m^{th}$ 'up' period ($Z_m$) if a type I change in the state of $X$ at time $T_{2m-1}$ is not announced to the control plane while the previous state ('0') was correctly perceived by the control plane; we show this event by $W^1_m=1$, and its probability is given by
\begin{eqnarray}
&Pr(W^1_m=1) = (1-U^1_X(\epsilon_1))Pr(\hat{S}_X=0|S_X=0)& \nonumber \\
&= (1-U^1_X(\epsilon_1))(\frac{Pr(S_X=0)-Pr(S_X=0,\hat{S}_X=1)}{Pr(S_X=0)}& \nonumber \\
&= (1-U^1_X(\epsilon_1))(1-D_2\frac{\tau_X+\theta_X}{\theta_X})&
\end{eqnarray}

In this case $\hat{S}_X=0$ during the time where $S_X=1$. So assuming that the $m^{th}$ such change is perceived wrong by the control plane, $Z_m$ is the time interval where the control plane has the type I wrong information about the state of $X$. Thus, the probability of type I error, and consequently type I distortion can be calculated as the ratio of total time of type I error over some time interval $[0,w]$ when $w\rightarrow \infty$.
\begin{eqnarray}
   D_1&=&E[\frac{1}{w}\sum_{m=1}^{N_w}1_{[W^1_m=1]}Z_m] \nonumber \\
	&=& \frac{1}{w}E[1_{[W^1_m=1]}Z_m]E[N_w] \nonumber \\
	&=&\frac{\tau_X}{\tau_X+\theta_X}Pr(W^1_m=1) \nonumber \\
	&=&\frac{\tau_X}{\tau_X+\theta_X}(1-U^1_X(\epsilon_1))(1-D_2\frac{\tau_X+\theta_X}{\theta_X})
\end{eqnarray}

Similarly, a false positive alarm is generated when a type II change is not announced while the previous perceived state ('1') was correct, and assuming that this is the $m^{th}$ such change, $Y_{m+1}$ is the time interval that the control plane has type II wrong information about $X$; let $W^2_m=1$ denote this event. Thus,
\begin{eqnarray}
&Pr(W^2_m=1) = (1-U^2_X(\epsilon_2))Pr(\hat{S}_X=1|S_X=1)& \nonumber \\
&= (1-U^2_X(\epsilon_2))\frac{Pr(S_X=1)-Pr(S_X=1,\hat{S}_X=0)}{Pr(S_X=1)} &\nonumber\\
&= (1-U^2_X(\epsilon_2))(1-D_1\frac{\tau_X+\theta_X}{\tau_X})\mbox{    and}& \nonumber\\
   &D_2=E[\frac{1}{w}\sum_{m=1}^{N_w}1_{[W^2_m=1]}Y_{m+1}] & \nonumber \\
	&= \frac{1}{w}E[1_{[W^2_m=1]}Y_{m+1}]E[N_w]=\frac{\theta_X}{\tau_X+\theta_X}Pr(W^2_m=1)& \nonumber \\
	&=\frac{\theta_X}{\tau_X+\theta_X}(1-U^2_X(\epsilon_2))(1-D_1\frac{\tau_X+\theta_X}{\tau_X})&\nonumber
\end{eqnarray}

To satisfy the distortion criteria we need $D_1\leq \epsilon_1$ and $D_2 \leq \epsilon_2$. The update rates $U^1_X(\epsilon_1)$ and $U^2_X(\epsilon_2)$ then can be written as
\begin{eqnarray}
&U^{1}_X(\epsilon_1) = 1-\frac{D_1 \frac{\theta_X}{\tau_X}}{\frac{\theta_X}{\tau_X+\theta_X}-D_2} \geq 1-\frac{\epsilon_1 \frac{\theta_X}{\tau_X}}{\frac{\theta_X}{\tau_X+\theta_X}-\epsilon_2}& \label{eq:U1} \\
&U^{2}_X(\epsilon_2) = 1-\frac{D_2\frac{\tau_X}{\theta_X}}{\frac{\tau_X}{\tau_X+\theta_X}-D_1} \geq 1-\frac{\epsilon_2\frac{\tau_X}{\theta_X}}{\frac{\tau_X}{\tau_X+\theta_X}-\epsilon_1}& \label{eq:U2}
\end{eqnarray}

Thus, the total number of updates announced to the control plane divided by the total number of changes is given by $U_X(\epsilon_1,\epsilon_2) = U^{1}_X(\epsilon_1) + U^{2}_X(\epsilon_2)$.

Note that the total rate of type I changes, which is equal to the rate of type II changes in average is given by $\frac{1}{\tau_X+\theta_X}$ changes per second, thus total number of updates per second is given by
\begin{eqnarray}
R_X(\epsilon_1,\epsilon_2) = \frac{U_X(\epsilon_1,\epsilon_2)}{\tau_X+\theta_X} \label{eq:RX}
\end{eqnarray}

Combining equations \ref{eq:U1}-\ref{eq:RX} the Lemma is proved.

\end{document}